\documentclass[journal,letterpaper,twocolumn,twoside,nofonttune]{IEEEtran}

\usepackage[utf8]{inputenc} 
\usepackage[T1]{fontenc}
\usepackage{url}
\usepackage{ifthen}
\usepackage{cite}
\usepackage{amsmath}%
\usepackage{wasysym}%
\usepackage{amssymb}

\usepackage{mdwmath}
\usepackage{blindtext}
\usepackage{eqparbox}
\DeclareMathOperator{\E}{\mathbb{E}}

\usepackage[english]{babel}
\usepackage{lipsum}
\usepackage{xcolor}
\usepackage{tikz}
\usepackage{graphicx}
\usepackage{algorithm, algorithmic}
\usepackage{siunitx}

\usepackage{calc}
\newcolumntype{M}[1]{>{\centering\arraybackslash}m{#1}}
\newcolumntype{N}{@{}m{0pt}@{}}

\usepackage{amsthm}

\usepackage{calc}
\newcolumntype{M}[1]{>{\centering\arraybackslash}m{#1}}
\newcolumntype{N}{@{}m{0pt}@{}}

\newtheorem{theorem}{{Theorem}}
\newtheorem{lemma}[theorem]{{Lemma}}
\newtheorem{proposition}[theorem]{{Proposition}}
\newtheorem{corollary}[theorem]{{Corollary}}

\newcommand{\cR}{{\cal R}}

\DeclareMathAlphabet{\mathbfsl}{OT1}{ppl}{b}{it} 

\newcommand{\be}[1]{\begin{equation}\label{#1}}
\newcommand{\ee}{\end{equation}} 
\newcommand{\eq}[1]{(\ref{#1})}

\renewcommand{\leq}{\leqslant}
 
\renewcommand{\geq}{\geqslant}

\renewcommand{\Bbb}{\mathbb}

\newcommand{\Tref}[1]{Theo\-rem\,\ref{#1}}
\newcommand{\Pref}[1]{Pro\-po\-si\-tion\,\ref{#1}}
\newcommand{\Lref}[1]{Lem\-ma\,\ref{#1}}
\newcommand{\Cref}[1]{Co\-ro\-lla\-ry\,\ref{#1}}

\newcommand{\Fq}{{{\Bbb F}}_{\!q}}

\IEEEoverridecommandlockouts
\begin{document}
	\setlength{\abovedisplayskip}{4pt}
	\setlength{\belowdisplayskip}{4pt}
\title{Coded Distributed Computing:\\ Performance Limits and Code Designs} 

 \author{
\IEEEauthorblockN{Mohammad Vahid Jamali, Mahdi Soleymani, and Hessam Mahdavifar\\}
   \IEEEauthorblockA{Department of Electrical Engineering and Computer Scienc, University of Michigan, Ann Arbor, MI 48109, USA\\
   }
   \IEEEauthorblockA{E-mails: \{mvjamali, mahdy, hessam\}@umich.edu 
   }
}


 \maketitle
\begin{abstract}
We consider the problem of coded distributed computing where a large linear computational job, such as a matrix multiplication, is divided into $k$ smaller tasks, encoded using an $(n,k)$ linear code, and performed over $n$ distributed nodes. The goal is to reduce the average execution time of the computational job. We provide a connection between the problem of characterizing the average execution time of a coded distributed computing system and the problem of analyzing the error probability of codes of length $n$ used over erasure channels. Accordingly, we present closed-form expressions for the execution time using binary random linear codes and the best execution time any linear-coded distributed computing system can achieve. It is also shown that there exist \textit{good} binary linear codes that attain, asymptotically, the best performance any linear code, not necessarily binary, can achieve. We also investigate the performance of coded distributed computing systems using
 polar and Reed-Muller (RM) codes that can benefit from low-complexity decoding, and superior performance, respectively, as well as explicit constructions. The proposed framework in this paper can enable efficient designs of distributed computing systems given the rich literature in the channel coding theory. 
\end{abstract}

\section{Introduction}
There has been increasing interest in recent years toward applying ideas from coding theory to improve the performance of various computation, communication, and networking applications. For example, ideas from repetition coding has been applied to several setups in computer networks, e.g., by running a request over multiple servers and waiting for the first completion of the request by discarding the rest of the request duplicates \cite{ananthanarayanan2013effective,vulimiri2013low,gardner2015reducing}. Another direction is to investigate the application of coding theory in cloud networks and distributing computing systems \cite{jonas2017occupy,lee2018speeding}.
A rule of thumb is that when the computational job consists of linear operations, coding techniques can be applied to improve the run-time performance of the system under consideration.

Distributed computing refers to the problem of performing a large computational job over many, say $n$, nodes with limited processing capabilities. A coded computing scheme aims to divide the job to $k<n$ tasks and then to introduce $n-k$ redundant tasks using an $(n,k)$ code, in order to alleviate the effect of slower nodes, also referred to as \textit{stragglers}. In such a setup, it is assumed that each node is assigned one task and hence, the total number of \textit{encoded tasks} is $n$ equal to the number of nodes. 

Recently, there has been extensive research activities to
leverage coding schemes in order to boost the performance of distributed computing systems \cite{lee2018speeding,li2016unified,reisizadeh2017coded,li2016coded,li2017coding,yang2017computing,lee2017high,dutta2016short,vulimiri2013low,wang2018coded,mallick2018rateless}. For example, \cite{lee2018speeding} has applied coding theory to combat the deteriorating effects of stragglers in matrix multiplication and data shuffling. The authors in \cite{reisizadeh2017coded} considered coded distributed computing in heterogeneous clusters consisting of servers with different computational capabilities.

Most of the work in the literature focus on the application of maximum distance separable (MDS) codes. However, encoding and decoding of MDS codes over real numbers, especially when the number of servers is large, e.g., more than $100$, face several barriers, such as numerical stability, and decoding complexity. In particular, decoding of MDS codes is not robust against unavoidable rounding errors when used over real numbers \cite{higham2002accuracy}. Employing large finite fields, e.g., coded matrix multiplication using \textit{polynomial codes} in \cite{yu2017polynomial}, can be an alternative approach. However, applying large finite fields imposes further numerical barriers due to quantization when used over real-valued data.

As we will show in Section III, MDS codes are \textit{theoretically} optimal in terms of minimizing the average execution time of any linear-coded distributed computing system. However, as discussed above,
their application comes with some practical impediments,
either when used over real-valued inputs or large finite fields, in most of distributed computing applications comprised of large number of local nodes. A sub-optimal yet practically interesting approach is to apply binary linear codes, consisting of $0$'s  and $1$'s, and then perform the computation over real values. In this case, there is no need for the quantization as a zero in the $(i,j)$-th element of the generator matrix of the binary linear code means that the $i$-th task is not included in the $j$-th encoded task sent to the $j$-th node while a one means it is included. To this end, in this paper, we consider $(n,k)$ binary linear codes where all computations are performed over real-valued data inputs. A related work to this model is the very recent work in \cite{bartan2019polar} where binary polar codes are applied for distributed matrix multiplication. The authors in \cite{bartan2019polar} justify the application of binary codes over real-valued data and provide a decoding algorithm using polar decoder.

In this work, we connect the problem of characterizing the average execution time of any coded distributed computing system to the error probability of the underlying coding scheme over $n$ uses of erasure channels (see \Lref{lemma1}). Using this connection, we characterize the performance limits of distributed computing systems such as the average execution time that any linear code can achieve (see \Tref{BLC}), the average job completion time using binary random linear codes (see \Cref{RC}), and the best achievable average execution time of any linear code (see \Cref{col5}) that can, provably, be attained using MDS codes requiring operations over large finite fields. Moreover, we study the  gap between the average execution time of binary random linear codes and the optimal performance (see \Tref{thm7}) showing the normalized gap approaches zero as $n \rightarrow \infty$ (see \Cref{col8}). This implies that there exist binary linear codes that attain, asymptotically, the best performance any linear code, not necessarily binary, can achieve. We further study the performance of coded distributed computing systems using
polar and Reed-Muller (RM) codes
that can benefit from low-complexity decoding and superior performance, respectively. 

\section{System Model}
We consider a distributed computing system consisting of $n$ local nodes with the same computational capabilities. The run time $T_i$ of each local node $i$ is modeled using a shifted-exponential random variable (RV), mainly adopted in the literature \cite{liang2014tofec,reisizadeh2017coded,lee2018speeding}. Then, when the computational job is equally divided to $k$ tasks, the cumulative distribution function (CDF) of $T_i$ is given by
\begin{align}\label{shifted-exp}
\Pr(T_i\leq t)=1-\exp\left(-\mu(kt-1)\right),~~~\forall ~\!t\geq 1/k,
\end{align}
where $\mu$ is the exponential rate of each local node, also called the straggling parameter. Using \eqref{shifted-exp} one can observe that the probability of the task assigned to the $i$-th server not being completed (equivalent to erasure) until time $t\geq1/k$ is
\begin{align}\label{epst}
\epsilon(t)\triangleq\Pr(T_i>t)=\exp\left(-\mu(kt-1)\right),
\end{align}
and is one for $t<1/k$. Therefore, given any time $t$, the problem of computing $k$ parts of the computational job over $n$ servers can be interpreted as the traditional problem of transmitting $k$ symbols, using an $(n,k)$ code, over $n$ independent-and-identically-distributed (i.i.d.) erasure channels. Note that the form of the CDF in \eqref{shifted-exp} suggests that $t_0\triangleq 1/k$ is the (normalized) deterministic time required for each server to process its assigned $1/k$ portion of the total job (all tasks are erased before $t_0$), while any time elapsed after $t_0$ refers to the stochastic time as a result of servers' statistical behavior (tasks are not completed with probability $\epsilon(t)$ for $t\geq t_0$).

Given a certain code and a corresponding decoder over erasure channels, a \textit{decodable} set of tasks refers to a pattern of unerased symbols resulting in a successful decoding with probability $1$. Then, $P_e(\epsilon,n)$ is defined as the probability of decoding failure over an erasure channel with erasure probability $\epsilon$. For instance, $P_e(\epsilon,1) = \epsilon$ for a $(1,1)$ code. Note that the reason to keep $n$ in the notation is to specify that the number of servers, when the code is used in distributed computation, is also $n$. Finally, the total job completion time $T$ is defined as the time at which a decodable set of tasks/outputs is obtained from the servers.

\section{Fundamental Limits}
The following Lemma connects the average execution time of {any} linear-coded distributed computing system to the error probability of the underlying coding scheme over $n$ uses of an erasure channel. 

\begin{lemma}\label{lemma1}
	The average execution time of a linear-coded distributed computing system using a given $(n,k)$ code can be characterized as
	\begin{align}
T_{\rm avg}\triangleq\E[T]&=\int_{0}^{\infty}P_e(\epsilon(\tau),n)d\tau\label{ET1}\\
&=\frac{1}{k}+\frac{1}{\mu k}\int_{0}^{1}\frac{P_e(\epsilon,n)}{\epsilon}d\epsilon\label{ET2},
\end{align}
where $\epsilon(\tau)$ is defined in \eqref{epst}.
\end{lemma}
\begin{proof}
It is well-known that the expected value of any RV $T$ is related to its CDF $F_T(\tau)$ as $\E[T]=\int_{0}^{\infty}(1-F_T(\tau))d\tau$. Note that $1-F_T(\tau)=\Pr(T>\tau)$ is the probability of the event that the job is not completed until some time $\tau$. Therefore, using the system model in Section II, we can interpret $\Pr(T>\tau)$ as the probability of decoding failure $P_e(\epsilon(\tau),n)$ of the code when used over $n$ i.i.d. erasure channels with the erasure probability $\epsilon(\tau)$. This completes the proof of \eqref{ET1}. Now given that for the shifted-exponential distribution $d\epsilon(\tau)/d\tau=-\mu k\epsilon(\tau)$, and that $P_e(\epsilon(\tau),n)=1$ for all $\tau\leq1/k$, we have \eq{ET2} by the change of variables.
\end{proof}

\noindent \textbf{Remark 1.} Note that \eqref{ET1} holds given any model for the distribution of the run time of the servers, while \eqref{ET2} is obtained under shifted-exponential distribution, with servers having a same straggling parameter $\mu$, and can be extended to other distributions in a similar approach. 

\begin{theorem}
\label{BLC}
	The average execution time of any linear-coded distributed computing system can be expressed as 
	\begin{align}\label{BL}
	T_{\rm avg}=\frac{1}{k}\left[1+\sum_{i=n-k+1}^{n}\frac{1}{i\mu}\right]+\frac{1}{\mu k}\sum_{i=1}^{n-k}\frac{1}{i}p(i,k),
	\end{align} 
	where $p(i,k)$
	is  the average conditional probability of decoding failure, for an underlying decoder, given $i$ encoded symbols are erased at random. 
\end{theorem}
\begin{proof}
Using the law of total probability and the definition of $p(i,k)$ we have
\begin{align}\label{BL-MAP}
P_e(\epsilon,n)=\sum_{i=1}^{n}\begin{pmatrix}
n\\i
\end{pmatrix}\epsilon^{i}(1-\epsilon)^{n-i}p(i,k).
\end{align}
Accordingly, characterizing $T_{\rm avg}$ requires computing integrals of the form $f_i\triangleq\int_{0}^{1}\epsilon^{i-1}(1-\epsilon)^{n-i}d\epsilon$ for $i=1,2,...,n$. Using part-by-part integration one can find the recursive relation $f_{i+1}=\frac{i}{n-i}f_i$ which results in $1/f_i={i\begin{pmatrix}
	n\\i
	\end{pmatrix}}$. Note that $p(i,k)=1$ for $i>n-k$, since one cannot extract the $k$ parts of the original job from less than $k$ encoded symbols. Then plugging \eqref{BL-MAP} into \eqref{ET2} leads to \eqref{BL}.
\end{proof}

Next, we characterize the average execution time using a random ensemble of binary linear codes with full-rank generator matrices. This random ensemble, denoted by $\cR(n,k)$, is obtained by picking entries of the $k \times n$ generator matrix independently and uniformly at random followed by removing those matrices not having a full row rank from the ensemble. 

\noindent\textbf{Remark 2.} Note that \eqref{BL-MAP} together with the integral form in \eqref{ET2} suggest that a coded computing system should always encode with a full-rank generator matrix, otherwise, the average execution time does not converge. This is the reason behind picking the particular ensemble described above. Note that this is in contrast with the conventional block coding, where we can get an arbitrarily small average probability of error over a random ensemble of all $k \times n$ binary generator matrices.

\begin{lemma}
\label{lemma2}
The probability that the generator matrix of a code picked from $\mathcal{R}(n,k)$ does not remain full row rank after erasing $i$ columns uniformly at random, denoted by $p_f(i,k)$, can be expressed as
\begin{align}\label{RCpe}
p_f(i,k)= 1-\frac{\prod_{j=1}^{k}\left(1-2^{j-1-n+i}\right)}{\prod_{l=1}^{k}\left(1-2^{l-1-n}\right)}.
\end{align}
\end{lemma}
\begin{proof}
Define $l(m,k)$ as the probability of $k$ binary uniform random vectors $\mathbf{v}_i\in\mathbb{F}_2^m$ being linearly independent. It is well-known that 
	\begin{align}\label{Pf}
	l(m,k)=\prod_{i=1}^{k}\left(1-2^{i-1-m}\right).
	\end{align}
Let $\tilde{\mathbf{G}}$ denote the  $k\times (n-i)$ matrix after removing $i$ columns of the $k\times n$ generator matrix $\mathbf{G}$ uniformly at random. 
Then
\begin{align}
p_f(i,k)&=\Pr\!\left(\big\{\mathrm{rank}(\tilde{\mathbf{G}})\neq k\big\}{\big |}\big\{\mathrm{rank}({\mathbf{G}})=k\big\}\right)\label{eq1}\\
&=1-\frac{\Pr\!\big(\big\{\mathrm{rank}(\tilde{\mathbf{G}})= k\big\}\big)}{\Pr\left(\big\{\mathrm{rank}({\mathbf{G}})=k\big\}\right)}\label{eq2}\\
&=1-\frac{\prod_{j=1}^{k}\left(1-2^{j-1-n+i}\right)}{\prod_{l=1}^{k}\left(1-2^{l-1-n}\right)},\label{eq4}
\end{align}
where \eqref{eq1} is by the definition of $p_f(i,k)$, \eqref{eq2} is by noting that $\Pr\!\big(\big\{\mathrm{rank}({\mathbf{G}})= k\big\}{\big |}\big\{\mathrm{rank}(\tilde{\mathbf{G}})=k\big\}\big)=1$, and \eqref{eq4} is by \eqref{Pf}.
\end{proof}
\begin{corollary}
\label{RC}
	The average execution time using random linear codes from the ensemble $\cR(n,k)$ under maximum a posteriori (MAP) decoding is given by \eqref{BL} while replacing $p(i,k)$ in \eqref{BL} by $p_f(i,k)$, characterized in Lemma 3.
\end{corollary}
\begin{proof}
The proof is by noting that the optimal MAP decoder fails to recover the $k$ input symbols given $n-i$ unerased encoded symbols if and only if the corresponding $k\times(n-i)$ sub-matrix of the generator matrix of the code is not full row rank which occurs with probability $p_f(i,k)$.
\end{proof}

\noindent\textbf{Remark 3.} \Tref{BLC} implies that the average execution time using linear codes consists of two terms. The first term is independent of the performance of the underlying coding scheme and is fixed given $k$, $n$, and $\mu$. However, the second term is determined by the error performance of the coding scheme, i.e.,  $p(i,k)$ for $i=1,2,...,n-k$, and hence, can be minimized by properly designing the coding scheme. 

The following corollary of \Tref{BLC} demonstrates that MDS codes, if they exist,\footnote{It is in general an open problem whether given $n$, $k$, and $q$, there exists an $(n,k)$ MDS code over $\Fq$ \cite[Ch. 11.2]{macwilliams1977theory}.} are optimal in the sense that they minimize the average execution time by eliminating the second term of the right hand side in \eqref{BL}. However, for a large number of servers $n$, the field size needs to be also large, e.g., $q > n$ for Reed-Solomon (RS) codes. 

\begin{corollary}[Optimality of MDS Codes]\label{col5}
For given $n$,$k$, and underlying field size $q$, an $(n,k)$ MDS code, if exists, achieves the minimum average execution time that can be attained by any $(n,k)$ linear code.
\end{corollary}
\begin{proof}
MDS codes have the minimum distance of $d_{\rm min}^{\rm MDS}=n-k+1$ and can recover up to $d_{\rm min}^{\rm MDS}-1=n-k$ erasures leading to $p(i,k)=0$ for $i=1,2,...,n-k$. Therefore, the second term of \eq{BL} becomes zero for MDS codes and they achieve the following minimum average execution time that can be attained by any $(n,k)$ linear code:
\begin{align}\label{MDS1}
T_{\rm avg}^{\rm MDS}=\frac{1}{k}+\frac{1}{\mu k}\sum_{i=n-k+1}^{n}\frac{1}{i}.
\end{align} 

\end{proof}

Using \Tref{BLC} and Remark 3, and given that the generator matrix of any $(n,k)$ linear code with minimum distance $d_{\rm min}$ remains full rank after removing up to any $d_{\rm min}-1$ columns, we have the following proposition for the \textit{optimality criterion} in terms of minimizing the average execution time.
\begin{proposition}[Optimality Criterion]\label{prop6}
An $(n,k)$ linear code that minimizes $\sum_{i=d_{\rm min}}^{n-k}{p(i,k)/i}$ also minimizes the average execution time of a coded distributed computing system.
\end{proposition}
Although MDS codes meet the aforementioned optimality criterion over large field sizes, to the best of our knowledge, the optimal linear codes, given the field size $q$ and in particular for $q=2$, per \Pref{prop6} are not known and have not been studied before, which calls for future studies.

In the following theorem we characterize the gap between the execution time of binary random linear codes and the optimal execution time. Then \Cref{col8} proves that binary random linear codes \textit{asymptotically} achieve the normalized optimal execution time, thereby demonstrating the existence of \textit{good} binary codes for distributed computation over real-valued data.
The reason we compare the normalized $nT_{\rm avg}$'s instead of $T_{\rm avg}$'s is that, using \eqref{BL}, $T_{\rm avg}$ has a factor of $1/k$ and hence, $\lim_{n\to\infty}T_{\rm avg}=0$ for a fixed rate\footnote{More precisely, the coding rate over field size $q$ is equal to $k\log_2 q/n$ but with slight abuse of terminology we have dropped the factor of $\log_2 q$ since this factor is not relevant for coded distributed computing.} $R\triangleq k/n >0$.
\begin{theorem}[Gap of Binary Random Linear Codes to the Optimal Performance]\label{thm7}
Let $T_{\rm avg}^{\rm BRC}$ denote the average execution time of a coded distributed computing system using binary random linear codes. Then, for any given $k$, $n$, we have
\begin{align}\label{asymp-gap}
&\frac{1}{3\mu R(1-R)n}< |nT_{\rm avg}^{\rm MDS}-nT^{\rm BRC}_{\rm avg}|<\frac{1}{\mu R}\times\nonumber\\
&\hspace{0.8cm}\left[\frac{v(n)}{n\!-\!k\!-\!v(n)\!+\!1}\!+\!nR2^{-v(n)}\ln\left(n\!-\!k\!-\!v(n)\right)\right]\!,
\end{align}
where $R$ is the rate and $v(n)$ is an arbitrary function of $n$ with $0 \leq v(n)\leq n-k$.
\end{theorem}
\begin{proof}
Using \Cref{RC} and \Cref{col5}, we have
\begin{align}\label{S}
    \mathcal{S}\triangleq\mu R|nT_{\rm avg}^{\rm MDS}-nT^{\rm BRC}_{\rm avg}|=\sum_{i=1}^{n-k}\frac{1}{i}p_f(i,k).
\end{align}
The lower bound in \eqref{asymp-gap} is by noting that 
$$
\mathcal{S}>p_f(n-k,k)/(n-k),
$$
where $p_f(n-k,k)$ can be expressed as
\begin{align}\label{lb1}
\!p_f(n\!-\!k,k)\!=\!1\!-\!\frac{1\!-\!2^{-k}}{1\!-\!2^{-n}}\cdot\frac{1\!-\!2^{-k+1}}{1\!-\!2^{-n+1}}\!\cdot\!...\!\cdot\frac{1-2^{-1}}{1\!-\!2^{k-n-1}}.\!
\end{align}
Note that $p_f(n-k,k)=0$ for $n=k$. For $n>k$, since $1-2^{-k+j}>1-2^{-n+j}$ for $j=0,1,...,k-2$, we 
have
\begin{align}
p_f(n-k,k)\!>\!1\!-\!\frac{1-2^{-1}}{1\!-\!2^{k-n-1}}\!>\!1\!-\!\frac{1-2^{-1}}{1\!-\!2^{-1-1}}\!=\!\frac{1}{3}.
\end{align}
Therefore, $\mathcal{S}>\frac{1}{3(1-R)n}$.

To prove the upper bound, the summation in \eqref{S} is split as $\mathcal{S}=\mathcal{S}_1+\mathcal{S}_2$ where
\begin{align}
    \mathcal{S}_1\triangleq\sum_{i=n-k-v(n)+1}^{n-k}\frac{1}{i}p_f(i,k)&<\frac{v(n)}{n-k-v(n)+1},\label{S1,1}
\end{align}
and 
\begin{align}\label{S2}
\mathcal{S}_2\triangleq\sum_{i=1}^{n-k-v(n)}\frac{1}{i}p_f(i,k).
\end{align}
To upper-bound $\mathcal{S}_2$, we first note that $p_f(i,k)$, defined in \eqref{RCpe}, is a monotonically increasing function of $i$. Then,
\begin{align}
\mathcal{S}_2&\leq p_f(n-k-v(n),k)\sum_{i=1}^{n-k-v(n)}\frac{1}{i}\label{S2,1}\\
&<p_f(n-k-v(n),k)\ln\left(n-k-v(n)\right)\label{S2,2}.
\end{align}
We can further upper-bound $p_f(n-k-v(n),k)$ as
\begin{align}
p_f(n-k-v(n),k)&<1-{\prod_{j=1}^{k}(1-2^{j-1-k-v(n)})}\label{pfinf1}\\
&<1-\left[1-2^{-v(n)}\right]^k\label{pfinf2}\\
&\leq nR2^{-v(n)}\label{pfinf3},
\end{align}
where \eqref{pfinf1} is by \eqref{RCpe} together with ${\prod_{l=1}^{k}\left(1-2^{l-1-n}\right)}\leq1$, \eqref{pfinf2} follows by noting that

$$\prod_{j=1}^{k}(1-2^{j-1-k-v(n)})={\prod_{j'=1}^{k}(1-2^{-j'-v(n)})}>[1-2^{-v(n)}]^k,$$

and \eqref{pfinf3} follows by Bernoulli's inequality $(1-x)^k\geq1-kx$ for any $0<x<1$ and then inserting $k=nR$.
\end{proof}

\begin{corollary}[Asymptotic Optimality of Binary Random Linear Codes]\label{col8}
The normalized average execution time $nT_{\rm avg}^{\rm BRC}$ approaches $nT_{\rm avg}^{\rm MDS}$ as $n$ grows large. More precisely, for a given rate $R$, there exist constants $c_1,c_2>0$ such that for sufficiently large $n$, i.e., $k=nR$, we have
\begin{align}\label{asymp-gap2}
c_1\frac{1}{n}\leq |nT_{\rm avg}^{\rm MDS}-nT^{\rm BRC}_{\rm avg}|\leq c_2\frac{\log_2 n}{n}.
\end{align}
\end{corollary}
\begin{proof}
The lower bound holds with $c_1=1/{3\mu R(1-R)}$ according to the left hand side of \eq{asymp-gap}. Observe that with the choice of $v(n)=2\log_2 n$ both terms in the right hand side of \eq{asymp-gap} become $O(\frac{\log_2 n}{n})$. Note that $n-k=n(1-R) \geq 2\log_2 n$, for sufficiently large $n$. Hence, the upper bound of \eq{asymp-gap2} also holds with a proper choice of $c_2$.  
\end{proof}

\noindent\textbf{Remark 4.}
 \label{optimumrate}
Using \eqref{MDS1} and a similar approach to \cite{lee2018speeding}, one can show that the asymptotically-optimal encoding rate  $R^*$ for an MDS-coded distributed computing system  is  the solution to 
 \begin{align}\label{RR}
(1-R^*)\ln(1-R^*)=\mu(1-R^*)-R^*.
\end{align}
 \Cref{col8} implies that for distributed computation using binary random  linear codes, the gap of $nT^{\rm BRC}_{\rm avg}$ to $nT^{\rm MDS}_{\rm avg}$ converges to zero as $n$ grows large. Accordingly, the optimal encoding rate also approaches $R^*$, described in \eqref{RR}.



\begin{table*}
     		\centering
     		\caption{Average execution time and optimal $k^*$ values for different coding schemes as well as their gap $g_{\rm opt}$ to the optimal performance and their performance improvement gain ${G}_{\rm cod}$ compared to the uncoded computing.}
     		\label{table}\vspace{-0.2cm}
   		\begin{tabular}{||M{0.09in}||M{0.65in}|M{0.8in}|M{1.1in}|M{1.1in}|M{1.1in}|M{1.1in}||}
     			\hline\hline
     \vspace{0.1cm}	$\!\!n$ & 
    {Uncoded}&
     {MDS coding}&
    {Binary random coding}&
     {Polar coding with SC}&
    {Polar coding with ML}&
    {RM coding with ML}\\ \cline{2-7}

    & $(T_{\rm avg},g_{\rm opt})$ &
     $(T_{\rm avg},k^*,{G}_{\rm cod})$ & 
     $(T_{\rm avg},k^*,g_{\rm opt},{G}_{\rm cod})$ & 
    $(T_{\rm avg},k^*,g_{\rm opt},{G}_{\rm cod})$ &
    $(T_{\rm avg},k^*,g_{\rm opt},{G}_{\rm cod})$
    &
    $(T_{\rm avg},k^*,g_{\rm opt},{G}_{\rm cod})$
    \\ \hline	\hline	
     			$\!\!8$  & $(0.4647,25\%)$& $(0.370,6,20\%)$&$(0.460,7,25\%,1.1\%)$&$(0.412,7,11\%,12\%)$&$(0.40,7,5.5\%,16\%)$&$(0.389,7,5.1\%,16\%)$\\\hline
     			$\!\!16$  & $(0.2738,44\%)$& $(0.191,11,31\%)$&$(0.226,11,18\%,18\%)$&$(0.217,11,14\%,21\%)$&$(0.199,11,4.2\%,28\%)$&$\!(0.198,11,3.6\%,28\%)$\\ \hline
     				$\!\!32$  & $(0.1581,63\%)$& $(0.0968,22,39\%)$&$(0.105,21,8.6\%,34\%)$&$(0.114,24,18\%,28\%)$&$(0.105,26,7.9\%,34\%)$&$\!(0.104,26,7.2\%,34\%)$\\ \hline
     					$\!\!64$  & $(0.0897,84\%)$& $(0.0488,44,46\%)$&$(0.051,43,3.9\%,44\%)$&$(0.0584,44,20\%,35\%)$&$\!(0.0533,46,9.4\%,41\%)$&$\!(0.050,42,2.6\%,44\%)$\\ \hline
     						$\!\!\!128$  & $\!\!(0.0503,105\%)$& $\!(0.0245,88,51\%)$&$\!(0.025,87,1.9\%,50\%)$&$\!(0.0293,88,19\%,42\%)$&$\!(0.0255,91,4.2\%,50\%)$&$\!(0.0252,97,2.8\%,50\%)$\\ \hline
     							$\!\!\!256$  & $\!\!(0.0278,127\%)$& $\!\!(0.0123,175,56\%)$&$\!\!(0.0124,174,0.9\%,56\%)$&$\!\!(0.0146,182,19\%,48\%)$&$\!\!(0.0129,186,5.5\%,54\%)$&$\!\!\!(0.0123,166,0.6\%,56\%)$\\ \hline
     								$\!\!\!512$  & $\!\!(0.0153,149\%)$& $\!\!\!(0.0061,350,60\%)$&$\!\!\!(0.0062,349,0.5\%,60\%)$&$\!\!(0.0073,388,19\%,52\%)$&$\!\!\!(0.0065,393,5.9\%,57\%)$&$\!\!\!(0.0061,353,0.1\%,60\%)$\\ \hline\hline
\end{tabular}
\vspace{-0.4cm}
\end{table*}

\section{Practical Codes and Simulation Results}
In this section, simulation results for the expected-time performance of various coding schemes over distributed computing systems are presented. In particular, their gap to the optimal performance are shown and also, their performance gains are compared with the uncoded computation. 
\subsection{Polar-Coded Distributed Computation}
Binary polar codes are capacity-achieving linear codes with explicit constructions and low-complexity encoding and decoding \cite{Arikan}. Also, the low-complexity $O(n\log n)$ encoding and decoding of polar codes can be adapted to work over real-valued data when dealing with erasures as in coded computation systems, as also noted in \cite{bartan2019polar}. Next, we briefly explain the encoding and decoding procedure of real-valued data using binary polar codes and delineate how we can obtain the average execution times using \Lref{lemma1}. 


\subsubsection{Encoding Procedure} Ar{\i}kan's $n\times n$ polarization matrix $\mathbf{G}_n=\begin{bmatrix}
1 & 0\\ 
1 & 1
\end{bmatrix}^{\otimes r}$ is considered, where $r=\log_2 n$ and $\mathbf{A}^{\otimes r}$ denotes the $r$-th Kronecker power of $\mathbf{A}$. Next, a design parameter $\epsilon_d$ is picked, as specified later in Section IV-C. Then the polarization transform $\mathbf{G}_n$ is applied to a binary erasure channel with erasure probability $\epsilon_d$, BEC$(\epsilon_d)$. The erasure probabilities of polarized bit-channels, denoted by $\{Z_i\}_{i=1}^n$, are sorted and the $k$ rows of $\mathbf{G}_n$ corresponding to the indices of the $k$ smallest $Z_i$'s are picked to construct the $k\times n$ generator matrix $\mathbf{G}$. The encoding procedure using the resulting $k\times n$ generator matrix $\mathbf{G}$, which also applies to any $(n,k)$ binary linear code operating over real-valued data, is as follows. First, the computational job is divided into $k$ smaller tasks. Then the $j$-th encoded task which will be sent to the $j$-th node, for $j=1,2,\dots,n$, is the sum of all tasks $i$'s for which the $(i,j)$-th element of $\mathbf{G}$ is $1$.
\subsubsection{Decoding Procedure} The recursive structure of polar codes can be applied for low-complexity detection/decoding of real-valued data using parallel processing for more speedups \cite{jamali2018low}. It is well-known that in the case of successive cancellation (SC) decoding over BECs, the probability of decoding failure of polar codes is $P_e^{\rm SC}(\epsilon,n)=1-\prod_{i\in\mathcal{A}}(1-Z_i)$, where $\mathcal{A}$ denotes the set of indices of the selected rows.


\noindent\textbf{Remark 5.} Since polar SC decoder is sub-optimal in terms of successful decoding performance, one can think of optimal maximum-likelihood (ML) decoder to attain a lower failure probability at the cost of higher complexities. Consequently, investigating the possibility of attaining close-to-ML performance, e.g., using SC list decoding of polar codes \cite{tal2015list}, over real-valued data is an interesting problem deserving future studies when taking into account all time-consuming components of a coded distributed computing system.
\subsubsection{Performance Characterization} Given the decoding method adopted we can find the average execution time using \Lref{lemma1}. In particular, when SC decoding is adopted, $T_{\rm avg}$ can be obtained by numerically evaluating the integral of \eqref{ET2} involving $P_e^{\rm SC}(\epsilon,n)$. Moreover, for the ML decoding, we first estimate the error probability $P_e^{\rm ML}(\epsilon,n)$ using Monte-Carlo (MC) simulations and then apply \eqref{ET2}. 

\subsection{RM-Coded Distributed Computation} 

RM codes are closely related to polar codes, where for an $(n,k)$ RM code the generator matrix $\mathbf{G}$ is constructed by choosing the $k$ rows of $\mathbf{G}_n$ (defined in Section IV-A1) having the largest Hamming weights. It is recently shown that RM codes are capacity achieving over BECs \cite{kudekar2017reed}, though under bit-MAP decoding, and numerical results suggest that they actually achieve the capacity with \textit{almost} optimal scaling \cite{hassani2018almost}. There is still a considerable interest in constructing low-complexity
decoding algorithms for RM codes attaining such performances. In this paper, we apply the MC-based simulation to estimate $P_e^{\rm ML}(\epsilon,n)$ for RM codes with the optimal ML decoder, and then evaluate their execution time, numerically, using \eqref{ET2}. The inspiration behind considering RM codes in this paper is that they are believed to have the \textit{almost} optimal scaling which, we conjecture, is sufficient for asymptotic optimality, similar to random linear codes in \Cref{col8}, for coded distributed computing. The simulation results, provided next, support this conjecture.


  \begin{figure}
	\centering
	\includegraphics[trim={0cm 0.5cm 0cm 0cm},width=3.6in]{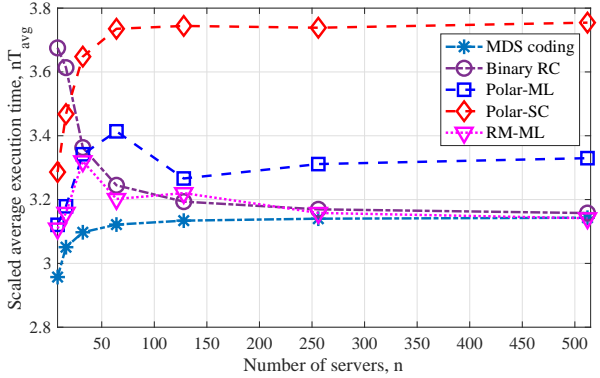}
	\caption{Scaled average execution time of a homogeneous distributed computing system with $\mu=1$ using various coding schemes for finite number of servers $n=8$, $16$, $32$, $64$, $128$, $256$, and $512$.}
	\label{F1}
\vspace{-0.15in}
\end{figure}
\begin{figure}
	\centering
	\includegraphics[trim={0cm 0.5cm 0cm 0cm},width=3.6in,height=2in]{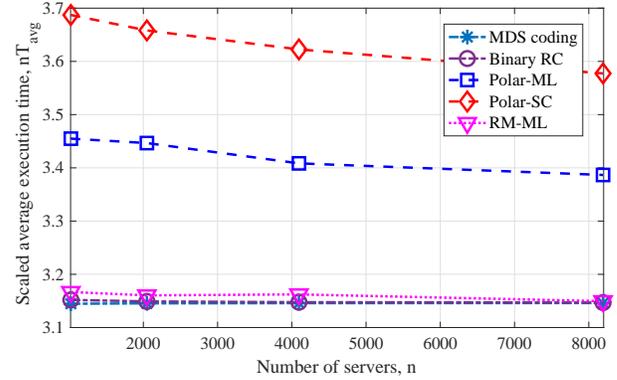}
	\caption{Scaled average execution time of a homogeneous distributed computing system with $\mu=1$ using various coding schemes for asymptotically large number of servers $n=1024$, $2048$, $4096$, and $8192$.}
	\label{F2}
\vspace{-0.15in}
\end{figure}
  
\subsection{Simulation Results}
 Numerical results for the performance  of the coded distributed computing systems utilizing MDS codes, binary  random linear codes, polar codes, and RM codes, are presented in Table \ref{table} and are compared with the uncoded scenario over  small block-lengths. We assume $\mu=1$ for all numerical results in this section. For MDS and random linear codes, $T_{\rm avg}$ is calculated  using  \eqref{MDS1} and \Cref{RC}, respectively,  and for polar and RM codes, it is  numerically evaluated using   \eqref{ET2} as discussed in Sections IV-A and IV-B. Then $k^*$ is obtained by minimizing $T_{\rm avg}$ for all possible values of $k$. We designed the polar code with  $\epsilon_d =0.1$, which is observed to be good enough for this range of block-lengths but one can also attain slightly better performance for  polar codes by optimizing over $\epsilon_d$ specifically for each $n$. Characterizing the best  $\epsilon_d$ as a function of  block-length $n$ is left for the future work. In Table I, $G_{\rm cod}$ is defined as the percentage of the gain in $T_{\rm avg}$ compared to the uncoded scenario and $g_{\rm opt}$ is defined as the gap of $T_{\rm avg}$ for  the underlying coding scheme  to that of MDS codes, in percentage. Intuitively, $G_{\rm cod}$ for a coding scheme determines \emph{how much gain} this scheme attains and $g_{opt}$ indicates  \emph{how close} this scheme is to the optimal solution. 
 Observe that polar codes with the low-complexity SC decoder achieve large enough $G_{\rm cod}$'s, close to the optimal values of $G_{\rm cod}$, e.g., $52\%$ for $n=512$ versus $60\%$ for the MDS code. Closer performance to the optimal $T_{\rm avg}$ can be obtained by decoding polar codes with ML decoder, e.g., $g_{\rm opt}=5.5\%$ for $n=256$.    
  Figure \ref{F1} shows that random linear codes have weak performance in the beginning  but they quickly approach the optimal $T_{\rm avg}$ so that they have small gaps to the optimal values, e.g., $g_{\rm opt}=0.5\%$ for $n=512$. Also, observe that RM codes always outperform polar codes since, perhaps, they have better distance distribution leading to better $p(i,k)$'s defined in \Tref{BLC}.
  
  In the case of $\mu=1$, by numerically solving \eqref{RR}, we have for the asymptotically-optimal encoding rate $R^*=0.6822$. 
 Motivated by this fact, in Figure 2, the rate of all discussed underlying coding schemes is fixed to $R^*$ and  $nT_{\rm avg}$ is plotted for moderately large block-lengths, i.e., $T_{\rm avg}$ is not optimized over  rates for the results demonstrated in this plot. Additionally,  the polar code is  designed with $\epsilon_d=1-R^*=0.3178$, which makes the code to be capacity-achieving for an erasure channel with capacity equal to $R^*$. Note that there is still a gap between  polar codes  with ML decoder and MDS codes. We believe this is due to the fact that binary polar codes with the $2 \times 2$ polarization kernel do not have an optimal scaling exponent \cite{hassani2014finite}. Furthermore, Figure \ref{F2} suggests that RM codes attain the optimal performance, and also do so relatively fast, supporting our conjecture in Section IV-B.

\section{Conclusion}

In this paper, we presented a coding-theoretic approach toward coded distributed computing systems by connecting the problem of characterizing their average execution time to the traditional problem of finding the error probability of a coding scheme over erasure channels. Using this connection, we provided results on the performance of coded distributed computing systems, such as their best performance bounds and asymptotic results using binary random linear codes. We further analyzed the performance of polar and RM codes in the context of distributed computing systems. We conjecture that achieving the capacity of BECs with optimal scaling exponent is a sufficient condition for binary codes to be asymptotically optimal, in the sense defined in \Tref{thm7}. We have shown this for binary random linear codes which are well-known to have optimal scaling exponent, even with sparse generator matrices \cite{mahdavifar2017scaling}, and numerically verified this for RM codes by observing that they attain close to optimal performance using a moderate number of servers. It is also interesting to see whether having an optimal scaling exponent is also a necessary condition for codes to be asymptotically optimal, e.g., whether binary polar codes with the $2 \times 2$ polarization kernel are asymptotically optimal or not.

\bibliographystyle{IEEEtran}

\bibliography{IEEEabrv}
\end{document}